\documentclass[a4paper,UKenglish]{article}
\usepackage[utf8]{inputenc}
\usepackage[T1]{fontenc}
 
\usepackage{microtype}%if unwanted, comment out or use option "draft"
\usepackage{xspace}
\usepackage{tikz}
\usepackage{multicol}
\usepackage{booktabs}

\graphicspath{{./img/}}%helpful if your graphic files are in another directory

\usepackage[left=2.5cm,right=2.5cm,top=2.5cm,bottom=3cm]{geometry}
\usepackage{amsmath,amsthm,amssymb}
\usepackage[pdfborder={0 0 0}]{hyperref}
\usepackage{colortbl}
\usepackage{booktabs}

\title{Computing Kernels in Parallel: Lower~and~Upper~Bounds}
\author{Max Bannach \and Till Tantau}
\date{%
  Institute for Theoretical Computer Science,\\
  Universit\"at zu L\"ubeck\\
  L\"ubeck, Germany \\
  \texttt{\{bannach,tantau\}@tcs.uni-luebeck.de}
}

\newtheorem{theorem}{Theorem}[section]
\newtheorem{corollary}[theorem]{Corollary}
\newtheorem{lemma}[theorem]{Lemma}

\theoremstyle{plain}
\newtheorem{fact}[theorem]{Fact}
\newtheorem{bremark}[theorem]{Remark}

\newtheorem{definition}[theorem]{Definition}

{%
  \leavevmode\nobreak\par
  \begin{list}%
    {}%
    {%
      \def\labelstyle{\itshape}
      \setlength{\topsep}{0pt}%
      \settowidth{\labelwidth}{\labelstyle Parameter:}%
      \setlength{\leftmargin}{\labelwidth}%
      \addtolength{\leftmargin}{\labelsep}%
      \setlength{\itemsep}{0pt}%
      \setlength{\parsep}{0pt}%
    }%
    }%
    {%
  \end{list}%
}

\newcommand\Class[1]{%
  \mathchoice%
  {\text{\normalfont\fontsize{9pt}{10pt}\selectfont$\mathrm{#1}$}}%
  {\text{\normalfont\fontsize{9pt}{10pt}\selectfont$\mathrm{#1}$}}%
  {\text{\normalfont$\mathrm{#1}$}}%
  {\text{\normalfont$\mathrm{#1}$}}%
}

\newcommand\Para{\mathrm{para\text-}}

\newcommand\Quasi{\mathrm{quasi\text-}}

\newcommand{\Lang}[1]{%
  \ifmmode{%
    \text{\normalfont\textsc{#1}}%
  }%
  \else
  {\normalfont\textsc{#1}}%
  \fi}

\def\tw{\mathrm{tw}}
\def\pw{\mathrm{pw}}
\def\td{\mathrm{td}}

\newcommand\PLang[1][]{p\def\test{#1}\ifx\test\stockhustantauempty\else_{\mathrm{#1}}\fi\text-\penalty15\Lang}
\newcommand\PLangText[1][]{$p\def\test{#1}\ifx\test\empty\else_{\mathrm{#1}}\penalty15\fi$-\penalty15\hskip0pt\textsc}

\begin{document}
\maketitle

\begin{abstract}
  Parallel fixed-parameter tractability studies how parameterized
  problems can be solved in parallel. A surprisingly large number of
  parameterized problems admit a high level of parallelization, but
  this does not mean that we can also efficiently compute
  small problem kernels in parallel: known kernelization
  algorithms are typically highly sequential. In the present paper,
  we establish a number of upper and lower bounds concerning the sizes
  of kernels that can be computed in parallel. An intriguing finding
  is that there are complex trade-offs between kernel size and the
  depth of the circuits needed to compute them: For the
  vertex cover problem, an exponential kernel can be computed by
  AC$^0$-circuits, a quadratic kernel by
  TC$^0$-circuits, and a linear kernel by randomized NC-circuits with
  derandomization being possible only if it is also possible for
  the matching problem. Other natural problems for which similar (but
  quantitatively different) 
  effects can be observed include tree decomposition problems parameterized
  by the vertex cover number, the undirected feedback vertex set
  problem, the matching problem, or the point line cover problem.  We also
  present natural problems for which computing kernels is inherently sequential.
\end{abstract}

\section{Introduction}
The core objective of parameterized complexity has classically been to
determine which problems can be solved in ``$\Class{FPT}$ time,''
meaning time $f(k)\cdot n^c$ for instances of size~$n$, where $c$~is a 
constant, $f$~is an arbitrary computable function (usually at least 
exponential), and $k$~is a hopefully small instance parameter. Over
the last 25 years, theoreticians in the field  
have been very successful at determining which problems admit
algorithms of this kind and practitioners have been very successful at
implementing them. In both cases, the focus has traditionally been
on finding \emph{sequential} algorithms, but in recent years interest
in \emph{parallel} algorithms has sparked, leading to the new field of
parallel fixed parameter tractability. 

In classical sequential $\Class{FPT}$ algorithms, \emph{kernelization
  algorithms} play a key role. They shrink the 
input to a small but difficult core (called the \emph{kernel}),
leading to the following design 
principle of modern parameterized algorithms:
Firstly, in polynomial time, a
kernelization algorithm computes a kernel that is, secondly, solved
using an exponential (or worse) time algorithm -- yielding a total running time
of the form $f(k)+n^c$. Regarding the parallelization of these two
algorithmic steps, it turns out that the second one is usually the
easier one: the kernel is typically processed using the search
tree technique or just by ``brute force,'' both of which allow natural
parallelizations. In contrast, kernelization algorithms are typically
described in a very sequential way, namely ``apply these reduction
rules over and over again.'' This means that designing parallel
fixed-parameter algorithms effectively means designing parallel
kernelization algorithms -- which is exactly what this paper
addresses.  

\subparagraph*{Our Contributions.}

We start our systematic investigation of parallel kernelization by
linking the parameterized analogues of the $\Class{NC}$-hierarchy to
kernel computation using $\Class{NC}$-circuits. Such a
link is already known for $\Class{FPT}$ and kernels computed in
polynomial time. We establish a circuit version of
the well-known result that all algorithms running in time $f(k)\cdot
n^c$ can also be implemented with running time $g(k)+n^c$: We can turn
any circuit family of size $f(k)\cdot n^c$ and depth $f(k)+c\log^i n$
into one of size $g(k)+n^{c'}$ and depth $c'\log^i n$ (note that we
can remove the parameter dependence from the depth).

The bulk of the paper consists of a series of lower and upper bounds
on the size of kernels that can be computed by circuits of certain depths.
We show that for natural problems like the vertex cover
problem intriguing trade-offs arise: the faster our algorithm,
the worse our kernel. For $\PLang{vertex-cover}$ we show that a simple
exponential kernel can be computed in $\Class{AC}^0$, a quadratic
kernel can be computed in $\Class{TC}^0$, and a linear kernel can be
computed in randomized $\Class{NC}$. Other problems for which we establish similar 
results include the  tree width, path width, and tree depth problems
parameterized by the vertex cover number of the input graph.

On the negative side, we also establish a number of lower bounds for
the parallel computation of small kernels. We show that a classical $2k$
kernel for the vertex cover problem can only be computed in parallel
if the maximum matching problem for bipartite graphs is in
$\Class{NC}$, for which $\Class{RNC}^2$ and $\Class{quasi\text-NC}^2$
are the best known upper bounds; that classic reduction rules for
feedback vertex set are $\Class{P}$-complete (but an exponential
kernel can be computed in $\Class{AC}^2$); that for the point line cover problem
we cannot (absolutely, without any assumptions) compute any
kernel in $\Class{AC}^0$ (but we can compute a quadratic one in $\Class{TC}^0$); and
that kernels for generalized versions of Horn satisfiability, linear
programming, and maximum flow cannot be computed in polylogarithmic
time unless $\Class{NC} = \Class P$. The later results in fact
presents three \emph{natural} $\Class{FPT}$-complete problems, which
demonstrate the limits of fixed parameter parallelization.

Table~\ref{table:overview} summarizes which trade-offs are established
in this paper between the parallel time needed to compute kernels and
their sizes. 

\begin{table}[htpb]
  \caption{An overview of problems studied in this paper,
    showing which kernel size can be achieved in certain layers of the
    $\Class{NC}$-hierarchy. An explicit function represents the best
    bound the authors are aware of, pointed out in this work or (for
    the $\Class{P}$-column) in cited
    works; $f(k)$ corresponds to kernels originating from
    Theorem~\ref{theorem:kernel}; and ``--'' means that there is no kernel
    of any size
    (either absolutely or unless $\Class{TC}^0=\Class{L}$ for~--$^1$,
    unless $\Class{TC}^0=\Class{NL}$ for~--$^2$,
    unless $\Class{TC}^0=\Class{P}$  for~--$^3$, unless
    $\Class{NC}=\Class{P}$ for~--$^4$, unless
    $\Class{P}\subseteq\Class{RNC}$ for~--$^5$, and unless
    $\Class{NC^1}=\Class{P}$ for~--$^6$). For problems
    parameterized by the vertex cover number, $S$ is the given vertex cover; the $\delta$ in the first column can be any
    fixed positive integer.
  }   
  \label{table:overview}
  \begin{center}
  \begin{tabular}{lccccc}
    \toprule
    Problem & &&\hbox to0pt{\hss Kernel size achievable in\hss}\\
    &$\Class{AC}^0$ & $\Class{TC}^0$ & $\Class{NC}$ & $\Class{RNC}$ & $\Class{P}$\\
    \cmidrule(rl){1-6}
    $\PLang{vertex-cover}$     & $2^{\sqrt[\delta]{k}}$   & $k^2+2k$ & $k^2+2k$ & $2k$ & $2k-c\log k$\\
    $\PLang{matching}$ & $2^{\sqrt[\delta]{k}}$ & $6k^2$ & $6k^2$ & $1$ & $1$\\
    $\PLang[vc]{tree-width}$   & $2^{\sqrt[\delta]{|S|}}$   & $|S|^3$    & $|S|^3$ & $|S|^3$ &  $|S|^3$\\
    $\PLang[vc]{path-width}$   & $2^{\sqrt[\delta]{|S|}}$   & $|S|^3$    & $|S|^3$ & $|S|^3$ &  $|S|^3$\\
    $\PLang[vc]{tree-depth}$   & $2^{\sqrt[\delta]{|S|}}$   & $|S|^3$    & $|S|^3$ & $|S|^3$ &  $|S|^3$\\
    $\PLang{point-line-cover}$ & --      & $k^2$    & $k^2$ & $k^2$ &  $k^2$\\
    $\PLang{feedback-vertex-set}$ & -- & --$^1$ & $f(k)$ & $f(k)$ & $2k^2+k$ \\
    $\PLang{strong-backdoor-2cnf-sat}$   & -- & --$^2$ & $f(k)$ &  $f(k)$ & $f(k)$\\
    $\PLang{strong-backdoor-horn-sat}$   & -- & --$^3$ & --$^3$ & --$^5$ & $f(k)$\\
    $\PLang{mixed-integer-programming}$   & -- & --$^6$ & --$^4$ & --$^5$ & $f(k)$\\
    $\PLang{max-flow-quantities}$   & -- & --$^6$ & --$^4$ & --$^5$ & $f(k)$\\
    \bottomrule
  \end{tabular}
  \end{center}
\end{table}

\subparagraph*{Related Work.}  Parameterized complexity is a rapidly
growing field, see \cite{Cygan:2015fr, DowneyF13, FlumG06} for an
introduction, in which parallelization is a recent
research direction. Early research in the late 1990s was done by Cai,
Chen, Downey, and  Fellows~\cite{CaiCDF97} who studied 
parameterized logaritmic space. A structural study of parameterized
logspace and parameterized circuit classes was started around 2015 by
Elberfeld et al.~\cite{ElberfeldST15}; see also the 
references therein. The 
parameterized version of the $\Class{NC}$-hierarchy we use in this
paper was introduced in~\cite{BannachST15}. Chen and
Flum studied lower bounds in this context and especially provide some
details and alternative characterizations for parameterized
$\Class{AC}^0$. There is a huge body of literature on polynomial-time
algorithms for computing small kernels, but the authors are not aware
of results concerning how quickly these kernels can be computed in
parallel. 

\subparagraph*{Organization of This Paper.}  We review basic
terminology in Section~\ref{section:terminology}, where we also
establish the link between parameterized parallel 
complexity and parallel kernel computation. Each of the following
sections studies a different well-known parameterized problem and
establishes trade-offs between kernel size and speed. We start with
the vertex cover and the matching problem in Section~\ref{section:vc},
followed by the feedback vertex set problem in
Section~\ref{section:fvs}, structural parameterizations for tree
width, path width, and tree depth in Section~\ref{section:structural}, the 
$\PLang{point-line-cover}$ problem in Section~\ref{section:tc0}, and
finally generalized versions of Horn satisfiability, linear
programming, and maximum flow in Section~\ref{section:seq}.

\section{Parameterized Parallel Complexity Classes and Kernelization}\label{section:terminology}
We use standard terminology of parameterized complexity theory, see
for instance~\cite{FlumG06}. A \emph{parameterized problem} is a tuple
$(Q,\kappa)$ consisting of a \emph{language} $Q\subseteq \Sigma^*$ and a
\emph{parameterization}
$\kappa\colon\Sigma^*\rightarrow\mathbb{N}$. The complexity of
$\kappa$ should not exceed the power of the classes that we consider, and
since we study small parameterized circuit classes, we require
$\kappa$ to be computable by \textsc{dlogtime}-uniform constant-depth
$\Class{AC}$-circuits or, equivalently, to be first-order computable.
 We denote parameterized problems by a leading
``$p$-'' as in $\PLang{vertex-cover}$, and, whenever the
parameterization~$\kappa$ is not clear from the context, we add it as an
index as in $\PLang[vc]{tree-width}$.
A parameterized problem $(Q,\kappa)$ is \emph{fixed-parameter
  tractable} (or in $\Class{FPT}$) if there is a computable function
$f\colon\mathbb{N}\rightarrow\mathbb{N}$ and a constant $c$ such that we
can decide $x\in Q$ in time $f(\kappa(x))\cdot |x|^c$
for all $x\in\Sigma^*$. In this paper we study the parallel
complexity of parameterized problems, that is, the parameterized
counter part of the $\Class{NC}$-hierarchy. Formally we study the
following classes, see for instance~\cite{BannachST15, ChenF16} for a
detailed discussion:
\begin{definition}\label{definition:paraac}
  For each $i > 0$, a parameterized problem $(Q,\kappa)$ is in
  \textsc{dlogtime}-uniform $\Para\Class{AC}^i$ if there exists a
  computable function $f\colon\mathbb{N}\rightarrow\mathbb{N}$, a
  constant $c\in\mathbb{N}$, and a family of $\Class{AC}$-circuits
  $(C_{n,k})_{n,k\in\mathbb{N}}$ such that:
\begin{enumerate}
\item For all $x\in\Sigma^*$ we have $C_{|x|,\kappa(x)}(x)=1\iff x\in Q$.
\item The depth of each $C_{n,k}$ is at most $f(k)+c\log^i n$. 
\item The size of each $C_{n,k}$ is at most $f(k)\cdot n^c$.
\item There is a deterministic Turing machine that on input of 
  $\operatorname{bin}(i)\#\operatorname{bin}(k)\#\operatorname{bin}(n)$,
  where $\operatorname{bin}(x)$ is the binary encoding of~$x$, outputs
  the $i$th bit of a suitable encoding of $C_{n,k}$ in at most $f(k)+c\log n$ steps. 
\end{enumerate}
\end{definition}
The class $\Para\Class{AC}^{0}$ is defined as above, but with circuits
of \emph{constant depth}. Additionally, we define for all $i\geq 0$
the class $\Para\Class{AC}^{i\uparrow}$ with
circuits of depth $f(k)\cdot\log^i n$. In particular, 
$\Para\Class{AC}^{0\uparrow}$-circuits have depth $f(k)$.
Recall that $\Class{AC}$-circuits are defined over the standard base
of \Lang{not}-, \Lang{or}-, and \Lang{and}-gates and that the last two may
have unlimited fan-in. The same definition works for
$\Class{NC}$-circuits (all gates have bounded fan-in) and
$\Class{TC}$-circuits (additional threshold gates are allowed).
It is known that the parameterized classes inherit their inclusion structure from
their classical counterparts~\cite{BannachST15}:
\[
  \Para\Class{AC}^0\subsetneq\Para\Class{TC}^0\subseteq\Para\Class{NC}^1\subseteq\Para\Class{AC}^1\subseteq\Para\Class{TC}^1
  \subseteq\dots\subseteq\Para\Class{NC}\subseteq\Class{FPT}.
\]

\subparagraph*{A Parallel Analogue of ``FPT = Kernels Computable in
  Polynomial Time''.}
One of the most fruitful aspects of parameterized complexity is the
concept of \emph{kernelization}. Let $f\colon\mathbb{N}\rightarrow\mathbb{N}$ be a
computable function. A \emph{kernelization} of a parameterized problem
$(Q,\kappa)$ is a self-reduction $K\colon\Sigma^*\rightarrow\Sigma^*$
such that for every $x\in\Sigma^*$ we have $x\in Q \iff
K(x)\in Q$ and $|K(x)|\leq f(\kappa(x))$. % folgt automatisch:, and $\kappa(K(x))\leq f(\kappa(x))$.
The images of $K$ are called \emph{kernels} and as they later need to
be processed by at least exponential-time algorithms, we are
interested in kernels that are as small as possible~-- while they still
need to be efficiently computable, meaning in polynomial time  from
the view point of $\Class{FPT}$ theory. The
following result is well-known and gives  
a deep connection between parameterized complexity and kernelization:
\begin{fact}[for instance~\cite{FlumG06}]
  A decidable parameterized problem $(Q,\kappa)$ is in $\Class{FPT}$ if, and only
  if, it admits a polynomial-time computable kernelization.
\end{fact}
The following theorem shows that the same relation also connects the
$\Class{AC}$-hierarchy with its parameterized counterpart. Note that
in the theorem the $\Class{AC}^i$-circuits are really ``normal
$\Class{AC}^i$-circuits,'' meaning that their size is just polynomial in
the input length.
\clearpage
\begin{theorem}\label{theorem:kernel}
  A decidable parameterized problem $(Q,\kappa)$ is in $\Para\Class{AC}^i$ if,
  and only if, it admits a kernelization computable by a
  \textsc{dlogtime}-uniform family of $\Class{AC}^i$-circuits.
\end{theorem}
\begin{proof}  
  Let $f\colon\mathbb{N}\rightarrow\mathbb{N}$ be a computable
  function and $c\in\mathbb{N}$ be a constant as in Definition~\ref{definition:paraac},
  let furthermore $(Q,\kappa)$ be a parameterized problem. Assume for the first
  direction that a kernelization $K$ of $(Q,\kappa)$ can be computed
  by a \textsc{dlogtime}-uniform family $(C_n)_{n\in\mathbb{N}}$ of
  $\Class{AC}^i$-circuits. Then we construct a family
  $(C_{n,k})_{n,k\in\mathbb{N}}$ of $\Para\Class{AC}^i$-circuits as
  follows: Circuit $C_{n,k}$ uses $C_n$ as a first black box, which is
  possible due to the depth and size definitions, and reduces the
  input to an instance of size at most $f(\kappa(x))$. Then the circuit
  essentially applies naive ``brute force'' in the form of a big
  \Lang{or}-gate that checks if any element of~$Q$ of length at most
  $f(\kappa(x))$ equals the computed kernel (we need the decidability
  of~$Q$ at this point to ensure that the circuit family is uniform). 

  For the other direction let us assume $(Q,\kappa)\in\Para\Class{AC}^i$
  witnessed by a \textsc{dlogtime}-uniform family
  $(C_{n,k})_{n,k\in\mathbb{N}}$ of $\Para\Class{AC}^i$-circuits, and
  let us first assume $i>0$.
  We may further assume that the computable function $f$ used in the definition
  of $(C_{n,k})_{n,k\in\mathbb{N}}$ is a monotone increasing function with
  $f(x)>x$ for all $x\in\mathbb{N}$, and that there is a Turing
  machine $M_f$ that computes $f(x)$ on input $\operatorname{bin}(x)$ in time
  $O(\log f(x))$. To see this, observe that, since $f$ is computable, there is a Turing machine
  $M'_f$ that computes $f(x)$ on input $\operatorname{bin}(x)$ in some
  time $T(x)$ such that $T$ is a monotone increasing function and such
  that $T(x)>x$ for all $x\in\mathbb{N}$. We may replace $f$ with $g(x)=2^{T(x)}$, which
  fulfils the above requirements. The resulting family
  $(C'_{n,k})_{n,k\in\mathbb{N}}$ is still a
  family of $\Para\Class{AC}^i$-circuits that accepts the same language.
  
  For every $n\in\mathbb{N}$
  we define $\tilde k\in\mathbb{N}$ to be the maximum $k$ such that
  $f(k)\leq c\log^i n$. We will use $\tilde k$ in the following
  construction and, hence, $\tilde k$ must be computable by a Turing
  machine in time $O(\log n)$ on input $\operatorname{bin}(n)$ to ensure
  uniformity. This is the case, as an appropriate Turing machine can
  first compute the value $c\log^i n$ (this is possible since $\log n$
  is a $\log\log n$-bit number, $i$ is a constant, and $c\log^i n$ is
  thus a $2^i\log\log n\in O(\log\log n)$-bit number) and can then
  perform binary search to find $\tilde k$. The later is possible
  since $f$ is monotone increasing and since $\tilde k \leq c\log^i n$
  as we have $f(x)>x$. Therefore, the Turing machine has to test only
  $\log\big(c\log^i n\big)\in O(\log\log n)$ possible $k$. Finally, for a fixed
  $k\leq c\log^i n$ the Turing machine can simulate $M_f$ on input
  $\operatorname{bin}(k)$ for $\log\big(c\log^i n\big)\in O(\log\log n)$
  steps and either obtains the value $f(k)$ or, if $M_f$ does not
  finish, can conclude that $f(k)>c\log^i n$.
  
  We now construct a family of $\Class{AC}^i$-circuits that compute a
  kernelization of $(Q,\kappa)$. Each circuit $C_n$ consists of
  $\tilde k$ subcircuits $C_n^0,\dots,C_n^{\tilde k}$ that are
  evaluated in parallel. The circuit $C_n^j$ first checks on input $x$
  whether or not $\kappa(x)=j$, which is possible since $\kappa$ can
  be computed by $\Class{AC}^0$-circuits (by definition). If this test
  is affirmative, the circuit uses $C_{n,j}$ to solve the problem and
  outputs a trivial kernel, that is, a trivial yes- or no-instance of
  $Q$. Otherwise $C_n^j$ just sets a flag that states that it is not
  responsible for this instance. Note that there is a constant $c'$
  such that $C_{n,j}$ has, by definition, depth at most
  $f(j)+c\log^i n\leq (c+1)\log^i n\leq c'\log^i n$ and size at most
  $f(j)\cdot n^c\leq c\log^i n \cdot n^c\leq n^{c'}$. If any $C_n^i$
  produces a kernel, then $C_n$ just presents this kernel as
  result. If, otherwise, all $C_n^j$ state that they are not
  responsible, we have $\kappa(x)>\tilde k$ and
  $f(\kappa(x))>c\log^i n$ and, thus, we already deal with a kernel,
  that is, $C_n$ can just present the input as output.

  For the remaining case, $i=0$, we perform the same
  construction, but choose $\tilde k$ such that $f(\tilde k)\leq n^c$,
  that is, we bound the subcircuits by size and not by depth.
\end{proof}

The theorem also holds if we replace $\Class{AC}^i$ with
$\Class{NC}^i$ or $\Class{TC}^i$. The only exception is
$\Class{NC}^0$, as this class may not be powerful enough to compute
$\kappa$. 

\subparagraph*{Application: Improve the Work of Parallel Algorithms.}
When we study the performance of parallel algorithms, we usually do
not only measure the time of the algorithm (as we would in the
sequential case), but also its \emph{work} (the total number
of computational steps performed by the algorithm). This is important
as a parallel algorithm may need polynomially many processors to reach
its promised runtime: For instance, an algorithm that runs in time $O(\log n)$
with $O(n^2)$ work will need at least time $O(n^2/p)$ on a machine
with $p$ processors~--~which is bad if there exists a linear time
sequential algorithm and $p < n$. In the circuit model the parallel time of an
algorithm corresponds to the depth of the circuit, and the work to its
size. While the layers of the $\Class{AC}$- and
$\Para\Class{AC}$-hierarchy measure the time of parallel algorithms
quite precisely, they only require the size of 
the circuits to be polynomial or to be bounded by $f(k)\cdot
n^c$, respectively. Using Theorem~\ref{theorem:kernel}, we can improve the work of any
parameterized parallel algorithm from $f(k)\cdot n^c$ to
$g(k)+n^{c'}$ while, at the same time, reducing the depth of the circuit
from $f(k)+c\log^i n$ to $c'\log^i n$.
\begin{lemma}\label{lemma:work}
  Let $(Q,\kappa)$ be a parameterized problem with
  $(Q,\kappa)\in\Para\Class{AC}^i$. Then there are a computable
  function $g\colon \mathbb{N}\rightarrow\mathbb{N}$ and a constant
  $c'$ such that there is a \textsc{dlogtime}-uniform family
  $(C'_{n,k})_{n,k\in\mathbb{N}}$ of $\Para\Class{AC}^i$-circuits that
  decides $(Q,\kappa)$ and in which every $C'_{n,k}$ has depth at most
  $c'\log^in$ and size at most $g(k)+n^{c'}$.
\end{lemma}
\begin{proof}
  Since $(Q,\kappa)\in\Para\Class{AC}^i$, there is a
  \textsc{dlogtime}-uniform family $(C_{n,k})_{n,k\in\mathbb{N}}$ of
  $\Para\penalty0\Class{AC}^i$-circuits that decides $(Q,\kappa)$. Let
  $f\colon\mathbb{N}\rightarrow\mathbb{N}$ and $c\in\mathbb{N}$ be as
  in Definition~\ref{definition:paraac}. By Theorem~\ref{theorem:kernel} there is a
  constant $c'$ and a \textsc{dlogtime}-uniform family
  $(C_n)_{n\in\mathbb{N}}$ of $\Class{AC}^i$-circuits such that every
  $C_n$ has depth at most $c'\log^i n$ and size at most $n^{c'}$ and produces a kernel of size at
  most $f(\kappa(x))$. 
  We construct the desired family $(C'_{n,k})_{n,k\in\mathbb{N}}$ as
  follows: The circuit $C'_{n,k}$ first applies the
  circuit~$C_n$ to an input~$x$ and obtains an instance~$x'$ of size at most
  $f(\kappa(x))$, then the circuit uses a constant number of
  $\Class{AC}$ layers to check $x'\in Q$ by testing in parallel for all $w\in Q$
  with $|w|\leq f(\kappa(x))$ whether $w=x'$ holds.

  Therefore, the depth of $C'_{n,k}$ equals (up
  to a constant) the depth of $C_n$, and the size of $C'_{n,k}$ is the
  sum of the size of $C_n$ and the size of the ``brute force'' circuit
  applied at the end, that is, there is a computable function
  $g\colon\mathbb{N}\rightarrow\mathbb{N}$ such that size of
  $C'_{n,k}$ can be bounded by $g(\kappa(x))+n^{c'}$.
\end{proof}

Note that the function $g$ from the lemma may grow exponentially
faster then $f$, as the circuit from the lemma internally solves an
instance $x'$ with $|x'|\leq f(\kappa(x))$ and
$\kappa(x')\leq f(\kappa(x))$. A direct application of
Lemma~\ref{lemma:work} is therefore only of theoretical interest. It
shows, however, that we can always search for parameterized parallel
algorithms that run in polylogarithmic time
and whose work is
polynomial plus an \emph{additive} term depending only on the parameter.

\section{Parallel Kernels for Vertex Cover and Matching}\label{section:vc}

The parameterized vertex cover problem is a prime example used to
demonstrate many different kernelization techniques, and an outrider in
the race for small kernels. In this section we
revisit the problem from the point of view of circuit complexity and
establish a link between circuit complexity and kernel size. 
An early result in this context is due to Cai et al.~\cite{CaiCDF97}
which, translated into the terminology of the present paper, implies
that a kernel for $\PLang{vertex-cover}$ can be computed in
logarithmic space and, hence, in $\Class{AC}^1$. Elberfeld et
al.~\cite{ElberfeldST15} later noticed that the kernel of size
$k^2+2k$ computed by Cai et al.\ can actually also be computed in
$\Class{TC}^0$. This result was later once more refined by showing that
the same kernel can be computed in
$\Para\Class{AC}^0$~\cite{BannachST15}. Together with Theorem~\ref{theorem:kernel}
this implies that a kernel of size $f(k)$ can be computed in
$\Class{AC^0}$ for some computable function $f$. In fact, we can
improve the bound in this case to $2^{\sqrt[\delta]{k}}$ for any fixed $\delta>0$:

\begin{lemma}\label{lemma:vcexp}
  For every $\delta\in\mathbb{N}$ there is a \textsc{dlogtime}-uniform
  family of $\Class{AC^0}$-circuits that, on input of a tuple $(G,k)$,
  outputs a $\PLang{vertex-cover}$ kernel with at most
  $2^{\sqrt[\delta]{k}}$ vertices.
\end{lemma}
\begin{proof}
  Let $I$ be the input instance and let $n=|I|$ be the size of its
  encoding. The circuit first checks if we have $k\leq\log^{\delta}(n)$. If
  not, we have $2^{\sqrt[\delta]{k}}>n$ and the instance  is already the desired kernel.
  Otherwise the circuit can simulate threshold gates up to
  $k$ using standard hashing techniques, as $\Class{AC^0}$-circuits
  can simulate polylogarithmic threshold
  gates~\cite{NewmanRW1990}. Since the $\Class{TC^0}$-circuit from
  Elberfeld et al.~\cite{ElberfeldST15} only uses threshold gates up
  to $k$, it follows that the $\Class{AC}^0$-circuit under construction
  can simulate this $\Class{TC^0}$-circuit, which completes the proof.
\end{proof}
The central observation in the proof of Lemma~\ref{lemma:vcexp} is
that the threshold-gates in the corresponding family of
$\Class{TC^0}$-circuits only ``count up to the parameter.'' We will
use exactly the same trick for other $\Class{TC}^0$-kernelizations,
but will then only formulate it as corollary.
Summarizing the statements from above, we can compute an exponential
kernel for $\PLang{vertex-cover}$ in $\Class{AC^0}$ and a quadratic
kernel in $\Class{TC^0}$. However, the best known kernelizations for
$\PLang{vertex-cover}$ are able to produce \emph{linear
  kernels}~--~and a reasonable next step is to implement them in
parallel as well. Unfortunately,  
this is
a way more challenging task, as both the classical $3k$ kernel based on
crown decomposition~\cite{Cygan:2015fr} and the $2k$ kernel due to Chen et
al.~\cite{ChenKJ01} require the computation of sufficiently large
matchings.  We can state this more precisely for the
latter  observation, by showing that the core part of the kernelization is
$\Class{NC}$-equivalent to computing maximum matchings in bipartite
graphs. The kernelization of Chen et al. is based on the following fact, known as the
Nemhauser--Trotter Theorem:
\begin{fact}[\cite{NemhauserT74}]\label{fact:nemhauser}
Let $G=(V,E)$ be a graph and $I=\{\,x_v\mid v\in V\,\}$ be a set of
variables. For every optimal solution $\beta\colon
I\rightarrow\mathbb{R}$ for the following linear program (\Lang{LPVC})
  \begin{align*}
    \min &\textstyle\sum\nolimits_{v\in V} x_v\\
    x_u+x_v&\geq 1 \quad \text{for all $\{u,v\}\in E$}\\
    x_v&\geq 0 \quad \text{for all $v\in V$}
  \end{align*}
  let
  $V_0=\{\,v\mid\beta(x_v)<1/2\,\}$,
  $V_{1/2}=\{\,v\mid\beta(x_v)=1/2\,\}$,
  $V_1=\{\,v\mid\beta(x_v)>1/2\,\}$
  be a partition of $V$. There is a minimum vertex cover $S$ of $G$ that
  satisfies $V_1\subseteq S\subseteq V_1\cup V_{1/2}$.
\end{fact}
Chen et al.~have shown that one can obtain the desired kernel from a
solution of \Lang{LPVC} by discarding the vertices of $V_0$ and by
taking the vertices of $V_1$ into the solution. The remaining $2k$
vertices of $V_{1/2}$ constitute the
kernel~\cite{ChenKJ01}. The following theorem shows that solving
\Lang{LPVC} is tightly linked to the maximum matching problems for bipartite graphs.
\begin{theorem}\label{theorem:lpvc}
  Computing a solution for \Lang{LPVC} is $\Class{NC}$-equivalent to
  computing a maximum matching in bipartite graphs.
\end{theorem}
\begin{proof}
The first direction is essentially the standard way of efficiently
solving \Lang{LPVC}: Given an instance of \Lang{LPVC} we construct a
bipartite graph $H=\bigl(\{\,v_1, v_2\mid v\in V\,\}, \bigl\{\{u_1, v_2\},\penalty0 \{u_2,
  v_1\}\mid\{u,v\}\in E\bigr\}\bigr)$ and compute a minimum vertex cover~$S$
  of it. One can show that the following assignment is an optimal
  solution for \Lang{LPVC}~\cite{Cygan:2015fr}:
  \[
    \beta(x_v)=\begin{cases}
      0           & \text{for } |\{v_1,v_2\}\cap S|=0,\\
      1/2         & \text{for } |\{v_1,v_2\}\cap S|=1, \text{ and}\\
      1           & \text{for } |\{v_1,v_2\}\cap S|=2.
    \end{cases}
  \]
  Since $H$ is bipartite, computing a minimum vertex cover is
  equivalent to computing a maximum matching due to König's
  Theorem~\cite{Koenig1916}. More precisely: To obtain the vertex
  cover~$S$, we compute a maximum matching in~$H$ and this matching
  constitutes 
  an optimal solution to the dual program of \Lang{LPVC}. Due to the
  Complementary Slackness Theorem, we can derive an optimal solution
  for the primal program from an optimal solution of the dual program
  by solving a linear system of equations, which is possible in
  $\Class{NC}$~\cite{JaJa92}. Note that the matrices of both
  \Lang{LPVC} and its dual are totally unimodular, as the incidence
  matrix of a bipartite graph is totally unimodular, and since the
  transpose of a totally unimodular matrix is so as well. Therefore,
  Cramer's Rule states that the solution that we obtain for the dual
  program with the algorithm from above is integral as
  well. This completes this part of the proof.

  For the other direction the input is a bipartite graph $G=(V,E)$ in
  which we search for a maximum matching. Let $\beta$ be an optimal real solution of
  \Lang{LPVC} for $G$. We can transform $\beta$ into a (still optimal)
  half-integral solution $\beta'$ by simple rounding:
  \[
    \beta'(x_v)=\begin{cases}
      0 & \text{if $\beta(x_v)<1/2$,}\\
      1/2 & \text{if $\beta(x_v)=1/2$, and}\\
      1 & \text{if $\beta(x_v)>1/2$.}\\
    \end{cases}
  \]
  This well-known fact is based on~\cite{NemhauserT74}, and can be
  shown by the following procedure that
  successively transforms $\beta$ into refined optimal solutions,
  ending at $\beta'$. To refine $\beta$ we define the sets
  $V_+=\{\,x_v\mid 0<\beta(x_v)<1/2\,\}$ and
  $V_-=\{\,x_v\mid 1/2<\beta(x_v)<1\,\}$. We now define for a
  suitable small $\epsilon>0$ the two
  assignments
  \[
    \beta_+(x_v)=\begin{cases}
      \beta(x_v) & \text{if $x_v\not\in V_+\cup V_-$,}\\
      \beta(x_v)+\epsilon & \text{if $x_v\in V_+$, and}\\
      \beta(x_v)-\epsilon & \text{if $x_v\in V_-$,}\\
    \end{cases}\quad\text{and}\quad
        \beta_-(x_v)=\begin{cases}
      \beta(x_v) & \text{if $x_v\not\in V_+\cup V_-$,}\\
      \beta(x_v)-\epsilon & \text{if $x_v\in V_+$, and}\\
      \beta(x_v)+\epsilon & \text{if $x_v\in V_-$.}\\
    \end{cases}
  \]
  Observe that both, $\beta_+$ and $\beta_-$, are still feasible
  solutions, as for any edge $\{u,v\}$ the constraint $x_u+x_v\geq 1$
  is still satisfied (either one of the variables is already $1$, or
  they are both $1/2$, or we add $\epsilon$ to at least one of
  them). Further observe that, compared to $\beta$, the value of the
  target function changes by $\epsilon|V_+|-\epsilon|V_-|$ and
  $\epsilon|V_-|-\epsilon|V_+|$, respectively. Since $\beta$ is
  optimal, neither $\beta_+$ nor $\beta_-$ may reduce the value of the
  target function compared to $\beta$; consequently we have
  $|V_+|=|V_-|$ and $\beta_+$ and $\beta_-$ are both optimal
  solutions. Conclusively observe that, by repeating this process
  successively, we will end up at $\beta'$.
  
  To conclude this part of the proof, we will now turn $\beta'$ into
  an integral solution. To achieve this, we construct
  an auxiliary graph $G'$ by deleting all vertices with value~$1$
  in~$G$ (as these must be in the vertex cover). Since all vertices
  with value~$0$ are now isolated, we may remove them too. We end up with a bipartite graph $G'$ with $n'$
  vertices, which are all assigned with the value $1/2$ by
  $\beta'$. We claim $\beta'$ is an optimal solution for \Lang{LPVC}
  on $G'$. For a contradiction assume otherwise, that is, assume there is
  an assignment $\gamma$ with $\sum_{v\in
    V(G')}\gamma(x_v)<\sum_{v\in V(G')}\beta'(x_v)$. We can infer a
  new assignment $\beta''$ for $G$ by ``plugging'' $\gamma$ into
  $\beta'$:
  \[
    \beta''(x_v)=\begin{cases}
      \beta'(x_v) & \text{if $x_v\not\in V(G')$;}\\
      \gamma(x_v) & \text{if $x_v\in V(G')$.}
    \end{cases}
  \]
  Observe that this is a feasible solution for \Lang{LPVC} on $G$,
  since for all edges $\{u,v\}$ we have:
  \[
    \beta''(x_u)+\beta''(x_v)=\begin{cases}
      \gamma(x_u)+\gamma(x_v)\geq 1 & \text{if $u,v\in V(G')$;}\\
      \beta'(x_u)+\beta'(x_v)\geq 1 & \text{if $u,v\not\in V(G')$;}\\
      \beta'(x_u)+\gamma(x_v)\geq 1 & \text{if $u\not\in V(G')$ and
        $v\in V(G')$.}\\
    \end{cases}
  \]
  The first two lines follow by the fact that $\gamma$ and $\beta'$
  are feasible; the last line follows by the construction of $G'$, as
  an edge $\{u,v\}$ with $u\not\in V(G')$ and $v\in V(G')$ only appears
  if we have $\beta'(x_u)=1$ (we have only deleted isolated vertices
  and vertices with value $1$, and here $u$ was deleted and is
  not isolated). By the construction of $\beta''$, we end up with $\sum_{v\in
    V(G)}\beta''(x_v)<\sum_{v\in V(G)}\beta'(x_v)$, which is a
  contradiction as $\beta'$ is an optimal solution for \Lang{LPVC} on
  $G$. Consequently, $\beta'$ must be an optimal solution for
  $\Lang{LPVC}$ on $G'$ as well.

  Since $\beta'$ assigns $1/2$ to all vertices in $G'$, a minimal
  vertex cover of $G'$ has size at least $n'/2$. Therefore, $G'$
  has to consist of two equally sized shores, as otherwise
  the smaller one would be a vertex cover of size smaller than
  $n'/2$. We can, thus, greedily select one shore into the vertex
  cover, that is, we set $\beta'$ for one shore to $1$ and for the
  other to $0$. The obtained optimal integral solution of
  $\Lang{LPVC}$ can be turned,
  as in the first direction, into a solution for the dual program in
  $\Class{NC}$, i.\,e., into a maximum matching of $G$.
\end{proof}
The parallel complexity of the maximum matching problem is still not
fully resolved. The currently best parallel algorithms run in
$\Class{RNC^2}$~\cite{MulmuleyVV87} or
$\Quasi\Class{NC^2}$~\cite{FennerGT16}. From the theorem we can deduce
that we can compute the Nemhauser--Trotter-based $2k$-vertex kernel for
$\PLang{vertex-cover}$ in $\Class{RNC}$ and $\Quasi\Class{NC}$; and we
can deduce
that we cannot compute this kernel in $\Class{NC}$ without
improving the parallel complexity of the maximum matching
problem~--~which is a longstanding open problem.
\begin{corollary}\label{corollary:vc2k}
  There is a \textsc{dlogtime}-uniform family of
  $\Class{NC}$-circuits of polylogarithmic depth that, on input of a
  graph $G=(V,E)$ and an integer $k$, outputs a kernel of
  $\PLang{vertex-cover}$ with at most $2k$ vertices. The circuits of
  the family either use randomness and have size $|V|^c$, or are
  deterministic and of size $|V|^{c\log |V|}$.
\end{corollary}
Note that other kernels that are based on the Nemhauser--Trotter
Theorem, such as the one by Soleimanfallah and
Yeo~\cite{SoleimanfallahY11}, or the one by Lampis~\cite{Lampis11}, also do not bypass
Theorem~\ref{theorem:lpvc}. A
natural goal is, thus, to compute linear kernels for  $\PLang{vertex-cover}$ in
$\Class{NC}$~--~most likely using an algorithm that does not rely on
a \Lang{LPVC} relaxation. Table~\ref{table:overview} summarizes the complexity
of computing kernels of certain size for $\PLang{vertex-cover}$.

Since $\PLang{matching}$ turns out to be an obstruction for parallel
kernelization, it is a natural question in the light of this paper,
whether or not we are able to compute polynomial kernels for the
matching problem in $\Class{NC}$. Note that the problem is in
$\Para\Class{AC}^0$, and hence we can compute a size-$f(k)$ kernel in
$\Class{AC^0}$; and since $\Lang{matching}\in\Class{RNC}$ we can
compute a size-1 kernel in $\Class{RNC}$.
\begin{lemma}\label{lemma:matchingTC}
  There is a \textsc{dlogtime}-uniform family of
  $\Class{TC^0}$-circuits that, on
  input of a tuple $(G,k)$, outputs a $\PLang{matching}$
  kernel with at most $O(k^2)$ vertices.
\end{lemma}
\begin{proof}
  The circuit first computes a set $S=\{\,v\in V\mid |N(v)|>2k\,\}$ of
  ``high-degree'' vertices. If we have $|S|\geq k$, the circuit
  can output a trivial yes-instance since for such a
  set~$S$ we can greedily match any vertex $v\in S$ with a vertex
  $u\in N(v)\setminus S$, reducing the available matching mates of
  all other vertices in $S$ by at most two~--~and since they have
  degree at least $2k$, there are still enough mates left to match every
  vertex of~$S$.

  If the circuit has not finished yet, we compute a set $S'$ consisting
  of~$S$ and $2k$ arbitrary neighbors of every vertex in~$S$ (take the
  lexicographically first for each $v \in S$, for instance). Note that
  we have $|S'|\leq 2k^2$. Consider the graph $G'=G[V\setminus
  S']$. Since $S$ was the set of high-degree vertices, $G'$ has
  maximum-degree $d\leq 2k$. Our circuit now removes
  all isolated vertices from~$G'$, resulting in~$G''$, and then checks
  if we have $|V(G'')|\geq k\cdot 2d$. If so, we can output a trivial
  yes-instance since a graph with  maximum degree~$d$ and minimum
  degree~$1$ always contains a matching of size $|V(G'')| / 2d \ge k$. If, on
  the other hand, we have $|V(G'')|\leq k\cdot 2d\leq 4k^2$, the
  circuit outputs $G[S'\cup V(G'')]$ together with the unchanged
  number~$k$. 

  The output clearly always has size at most $O(k^2)$. To see that
  $G[S' \cup V(G'')]$ is a kernel, we clearly only have to show that
  if $G$ has a size-$k$ matching~$M$, so does $G[S' \cup V(G'')]$ (the
  other direction is trivial). To see this, first note that any edge
  in~$M$ that does not have an endpoint in $S$ must lie in $G''$ and,
  hence, is also present in $G[S' \cup V(G'')]$. Next, all other
  edges in~$M$ must have an endpoint in~$S$ and, thus, there can be at
  most $|S|$ many such edges. While not all of these edges need to be
  present in $G[S']$, we can greedily construct a matching of size
  $|S|$ in $G[S']$ (by the same argument as the one of the beginning
  of this proof for $|S| \ge k$). This means that we find a matching
  of size $|M|$ also in $G[S' \cup V(G'')]$.
\end{proof}
The circuits of Lemma~\ref{lemma:matchingTC} need their threshold gates
``only'' to count up to $k$. We can thus deduce the
following corollary (the proof argument is the same as for Lemma~\ref{lemma:vcexp}):
\begin{corollary}
  For every $\delta\in\mathbb{N}$ there is a \textsc{dlogtime}-uniform family of
  $\Class{AC^0}$-circuits that, on
  input of a tuple $(G,k)$, outputs a $\PLang{matching}$
  kernel with at most $O(2^{\sqrt[\delta]{k}})$ vertices.
\end{corollary}

\section{Parallel Kernels for the Feedback Vertex Set Problem}\label{section:fvs}

The input for $\PLang{feedback-vertex-set} = \PLang{fvs}$ is an undirected multigraph
$G=(V,E)$ and an integer~$k$, the question is whether it is possible
to delete $k$ vertices such that the remaining graph is a forest. The
problem is well-known to be fixed-parameter tractable. Concerning the
parallel complexity, it is known that membership in $\Class{FPT}$ can
be witnessed by a machine that uses ``$\Class{FPT}$ time and $\Class{XL}$
space''~\cite{ElberfeldST15} and the problem was recently shown to lie in
$\Para\Class{NC}^{2+\epsilon}\subseteq\Para\Class{NC^3}$~\cite{BannachT16}.

A lot of effort has been put into the design of sequential kernels for
this problem, ultimately resulting in a kernel with  
$O(k^2)$ vertices~\cite{BurrageEFLMR06, BodlaenderD10,
  Thomasse10, Iwata17}. Much less is known concerning parallel
kernels. Since the $k=0$ slice of $\PLang{fvs}$ is exactly the $\Class
L$-complete~\cite{CookM87} problem whether a given graph is a forest,
we get as a lower bound that no kernel of any size can be computed for
$\PLang{fvs}$ by any circuit class~$C$ unless $\Class{L} \subseteq C$
and the smallest $\Class{AC}$-class for which this is known is
$\Class{AC}^1$. On the other hand, the mentioned membership in
$\Para\Class{NC}^{2+\epsilon}$ 
together with Theorem~\ref{theorem:kernel} yield an $\Class{NC}^{2+\epsilon}$
kernel. In summary:
\begin{lemma}\label{lemma:fvs-up-down}
  There is a \textsc{dlogtime}-uniform family of
  $\Class{NC}^{2+\epsilon}$-circuits that, on
  input of a tuple $(G,k)$, outputs a $\PLang{fvs}$
  kernel with at most $f(k)$ vertices. There is no such family of
  $\Class{AC}^{1-\epsilon}$-circuits, unless $\Class L \subseteq \Class{AC}^{1-\epsilon}$.
\end{lemma}
A natural first question arising from this lemma is: Can we improve
the bounds? It turns out that we can lower the upper bound from
$\Class{NC}^{2+\epsilon}$ to $\Class{AC}^{1+\epsilon}$ by observing the 
reduction rules used in sequential kernels for $\PLang{fvs}$ can, in
certain cases, be applied in parallel. In detail, the known sequential
kernels for $\PLang{fvs}$ all repeatedly apply 
(at least) the below rules, whose correctness is very easily
seen. We will show that each of the first three rules can individually
be applied exhaustively in $\Class{AC^1}$. Based on this, we show
$\PLang{fvs}\in\Para\Class{AC}^{1\uparrow}$.  
\begin{description}
  \def\myskip{\hphantom{Flower Rule}}
\item[\rlap{Leaf Rule}\myskip] Delete a vertex $v$ of degree $1$.
\item[\rlap{Chain Rule}\myskip] Contract a vertex $v$ of degree $2$ to one of its
  neighbors.
\item[\rlap{Loop Rule}\myskip] Delete a vertex $v$ with $v\in N(v)$, reduce $k$ by~$1$.
\item[\rlap{Flower Rule}\myskip] Delete a vertex~$v$ that appears in more then $k$
  cycles that only share the vertex~$v$, reduce $k$ by~$1$.
\end{description}
\clearpage
\begin{lemma}\label{lemma:fvs:rulesInAC}
  There is a \textsc{dlogtime}-uniform family of
  $\Class{AC^1}$-circuits that, on
  input of a tuple $(G,k)$, outputs a tuple $(G',k')$ that results
  from repeatedly applying (only) the Leaf Rule as long as possible. The same holds for
  the Chain Rule and for the Loop Rule.
\end{lemma}
\begin{proof}
  The claim follows immediately for the Loop Rule as we may delete all
  such vertices in parallel and since the deletion of a vertex cannot
  create new vertices with a self-loop. For the other two rules
  observe that an ``exhaustive application'' equals either the
  deletion of attached trees (for the Leaf Rule), or the contraction
  of induced paths (for the Chain Rule). For the first case, the
  circuit must be able to detect if a vertex $v$ becomes a leaf at
  some point of the computation (of course, the circuit cannot sequentially delete
  degree one vertices). The following observation provides a locally
  testable property that allows precisely such a detection: A vertex
  $v$ is contained in an attached tree if, and only if, it is possible to delete a
  single edge such that (a) the graph decomposes into two components
  and such that (b) the component of $v$ is a
  tree~\cite{ElberfeldST15}. Both properties can be tested in logspace
  (and hence in $\Class{AC}^1$), and an $\Class{AC}^1$-circuit can
  test them for all vertices and all edges in parallel. Finally, for
  the Chain Rule, observe that an $\Class{AC}^1$-circuit can mark all degree two vertices in parallel and that such a circuit,
  afterwards, only has to connect the two endpoints of highlighted
  paths~--~which is a again a logspace task.
\end{proof}
\begin{theorem}\label{theorem:fvs:ac1up}
  $\PLang{fvs}\in\Para\Class{AC}^{1\uparrow}$.
\end{theorem}
\begin{proof}
  We have to construct a family of $\Class{AC}$-circuits of depth
  $f(k)\cdot\log n$ and size $f(k)\cdot n^c$. The circuits will
  consist of $k$ layers such that every layer finds a set of at
  most $3k$ vertices to branch on (which will be done for the next
  layer). Note that layer~$i$ contains at most $3k$ as many
  subcircuits as layer $i-1$.

  Each layer consists of multiple $\Class{AC^1}$-circuits that work
  independently of each other on different possible graphs (depending
  on the branches of the previous layer). Each of these circuits first
  checks if the input is a yes-instance (input
  is a tree and $k\geq 0$), or a no-instance ($k<0$)~--~in the first case it just globally
  signals this circumstance and in the second case it truncates this
  path of the computation. If the subcircuit has not decided yet, it
  applies first the Leaf Rule exhaustively, and then Chain Rule
  exhaustively~--~both are possible due to
  Lemma~\ref{lemma:fvs:rulesInAC}. The circuit now applies the
  Loop Rule (again, using Lemma~\ref{lemma:fvs:rulesInAC}), if the
  rule has an effect (that is, $k$ was reduced by at least one) the
  circuit is done and just pipes the result to the next layer. If not,
  the circuit tests in parallel if there are two vertices $v$ and $u$
  that are connected by a multi-edge (that is, by at least two
  edges). If this is the case, any feedback vertex set must contain
  either $v$ or $u$ and, hence, the circuit branches on these two
  vertices and pipes the two resulting graphs to the next
  layer. Otherwise, we know that we have no vertex with a self-loop,
  no vertices with multi-edges, and a minimum vertex-degree of at
  least three. The circuit then uses the simple fact that any size $k$
  feedback vertex set in such graph must contain at least one vertex
  of the $3k$ vertices of highest degree, and hence, may simple branch
  over these~\cite{Cygan:2015fr}.

  Since each layer reduced $k$ in each branch by at least one, after
  at most $k$ layers every branch has decided if it deals with a yes-
  or a no-instance. Since each layer is implemented by an
  $\Class{AC^1}$-circuit, the claim follows.
\end{proof}
\begin{corollary}
  There is a \textsc{dlogtime}-uniform family of
  $\Class{AC}^{1+\epsilon}$-circuits that, on
  input of a tuple $(G,k)$, outputs a $\PLang{fvs}$
  kernel with at most $f(k)$ vertices.
\end{corollary}
\begin{proof}
  Follows by Theorem~\ref{theorem:kernel} and by the fact that $\Para\Class{AC}^{i\uparrow}\subseteq\Para\Class{AC}^{i+\epsilon}$~\cite{BannachST15}.
\end{proof}

We now have rather tight bounds (an upper bound of
$\Class{AC}^{1+\epsilon}$ and a conditional lower bound of
$\Class{AC}^{1-\epsilon}$) on how quickly we can compute 
\emph{some} kernel for $\PLang{fvs}$ in parallel. However, there is a
natural second question arising from Lemma~\ref{lemma:fvs-up-down}: Can we also compute a polynomial kernel in parallel?

We claim that progress towards such a kernel cannot
solely be based on the presented reduction rules.  
In the proof of Theorem~\ref{theorem:fvs:ac1up} we may need to branch
after the exhaustive application of one of the rules 
Leaf Rule, Chain Rule, or Loop Rule. If we seek to implement
a polynomial kernel for $\PLang{fvs}$ in $\Class{NC}$, we have to
implement these rules without branching and have to apply the
rules exhaustively \emph{together} while they may influence
each other. Figure~\ref{figure:fvs} provides an intuition why this interplay 
is ``very sequential,'' and Theorem~\ref{theorem:fvs:ruleP}
provides evidence that it is in fact very unlikely that there exists a
parallel algorithm that computes the result of jointly applying all rules exhaustively.

\clearpage
\begin{figure}[ht]
  \begin{center}
    \includegraphics{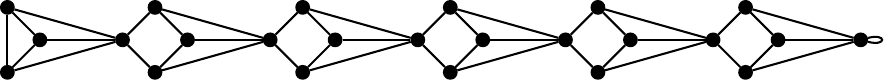}     
  \end{center}
  \caption{A graph that is fully reduced by the Chain Rule and the
    Loop Rule in $k=6$ rounds. In every round, the Chain Rule can only
    be applied after the Loop Rule was used exhaustively.}
  \label{figure:fvs}
\end{figure}
\begin{theorem}\label{theorem:fvs:ruleP}
  The problem of deciding whether a specific vertex of a given graph
  will be removed by an exhaustive application of the Leaf Rule, the
  Chain Rule, and the Loop~Rule is $\Class{P}$-hard under $\Class{NC^1}$-reduction. 
\end{theorem}
\begin{proof}
  We will reduce from the 
  monotone circuit value problem (\Lang{MCVP}), which is known to be 
  $\Class{P}$-complete under
  $\Class{NC}^1$-reduction~\cite{Greenlaw95}. The input to this
  problem is 
  a monotone circuit (it consists only of \Lang{and}-gates and \Lang{or}-gates of
  indegree~$2$, and it has a single gate marked as output) and an assignment of the input gates, the question is
  whether or not the output gate evaluates to
  true. We will transform the input circuit into a
  multi-graph by replacing any gate with a small gadget. Every gadget
  will have two vertices marked as ``input'' and one marked as
  ``output''. The ``input'' vertices are incident to exactly one edge
  outside of the gadget (which connects them to the ``output'' vertex
  of another gadget), the ``output'' vertex of the gadget may have edges
  to an arbitrary number of other ``input'' vertices.
  The semantic then is as follows: The edge of an ``input'' vertex
  that leaves the gadget will be removed by the reduction rules 
  when the corresponding wire of the
  circuit would have the value true for the given assignment of
  the input gates; similarly the ``output'' vertex of the gadget will
  be removed if, and only if, the corresponding gate would evaluate to
  true under the given assignment (this in turn removes the edges to
  other ``input'' vertices and propagates the computation of the circuit). By induction, the ``output'' vertex of the gadget
  corresponding to the output gate will then be removed if, and only
  if, the circuit evaluates to true~--~which completes the proof.

  We start with a description of the transformation. For
  clarity, we stipulate that a self-loop contributes two to the
  degree of a vertex, similarly multi-edges increase the degree by
  their multiplicity. Therefore, the Chain Rule may not be applied to a
  leaf with a self-loop. We further stipulate that the Chain Rule may
  not be applied to a self-loop, i.\,e., it has to contract two
  distinct vertices (and hence, self-loops may only be handled by the
  Loop Rule). 
  For the input gates, we use the gadgets
  \raisebox{-.3ex}{\includegraphics{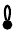}}
  and
  \raisebox{-3.55508pt}{\includegraphics{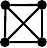}}
  for gates with assignment true or false, respectively. For
  \Lang{and}-gates, we use the gadget
  \hspace{-1ex}\raisebox{-2ex}{\includegraphics{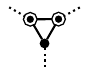}}\hspace{-1ex},
  and for \Lang{or}-gates
  \hspace{-1ex}\raisebox{-2ex}{\includegraphics{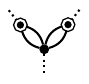}}\hspace{-1ex}.
  In these figures, the two top circled vertices are
  the ones we call ``input'', while the bottom vertex is the
  ``output'' vertex. The dotted lines indicate edges that leave the
  gadget. For every ``input'' vertex there will be exactly one
  outgoing edge, as any gate has exactly to incoming wires. The
  ``output'' vertex may have edges to an arbitrary number of successor
  gates; to ensure that there is at least some edge, we fully connect such
  vertices to cliques of size three (that is, the ``output'' vertex is
  part of a clique of size four)~--~this ensures that the degree of
  ``output'' vertices is always greater then two.

  We first prove that these gadgets work locally as
  intended, that is, that they perform as an input gate that is set to
  true or   false, or as an \Lang{and}- or \Lang{or}-gate. Observe
  that the gadget for input gates that are true contains a self-loop,
  that is, it gets removed by 
  the Loop~Rule; and observe that the gadget for a false input gate is
  a clique to which no rule can be applied. Now for the internal
  gates, observe that all vertices have degree at least three and no
  self-loop, that is, no rule can be applied unless one ``incoming'' edge
  gets removed (recall that each ``output'' vertex is fully connected to a
  clique of size three). In other words, the
  gadgets simulate the corresponding gates correctly for the
  assignment ``(false, false)'', a case distinction shows that this is
  also the case for the other assignments:
  \begin{center}
    \def\myskip{\vspace{0.55ex}}
    \begin{tabular}[t]{ccc}
      \toprule
      Assignment &\hbox to0pt{ Behaviour of the \hss}\\
      & \Lang{and}-gadget & \Lang{or}-gadget\\
      \cmidrule(rl){1-3}
      (true, false) &
\(
  \raisebox{-2ex}{\includegraphics{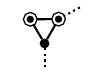}}\mapsto
  \raisebox{-2ex}{\includegraphics{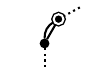}}
\) & \(
     \raisebox{-2ex}{\includegraphics{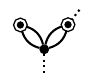}}\mapsto
     \raisebox{-2ex}{\includegraphics{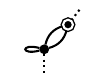}}
     \)\\
      (false, true) & \(
   \raisebox{-2ex}{\includegraphics{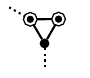}}\mapsto
   \raisebox{-2ex}{\includegraphics{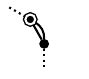}}
 \) & \(
 \raisebox{-2ex}{\includegraphics{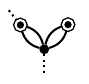}}\mapsto
 \raisebox{-2ex}{\includegraphics{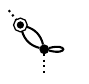}}                      
\)\\
      (true, true) & \(
\raisebox{-2ex}{\includegraphics{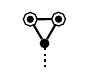}}\mapsto
\raisebox{-2ex}{\includegraphics{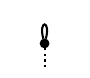}}
\) &  \(
\raisebox{-2ex}{\includegraphics{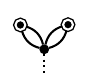}}\mapsto
\raisebox{-2ex}{\includegraphics{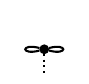}}
 \)\\  
      \bottomrule
    \end{tabular}
  \end{center}
  Observe that the ``output'' vertex obtains a self-loop (and hence
  gets removed by the Loop~Rule) if, and only
  if, the corresponding gate evaluates to true; note that the degree
  of a ``output'' vertex is always greater than two and,
  hence, it is never affected by the Leaf Rule or the Chain Rule.

  We now show the
  correctness of the construction by an induction over the gates of
  the input circuit in topological order. The induction hypothesis is
  that the gadget corresponding to the current gate gets modified by
  the Leaf Rule, the Loop Rule and the Chain Rule in the same way as the gate gets
  evaluated. The base case is given as this is true for the input
  gates by construction. For the inductive step consider the gadget
  corresponding to any gate $g$, and let it have the vertices $x$, $y$,
  and $z$ where $z$ is the ``output'' vertex. By the induction hypothesis the vertices
  marked as ``input'' (i.\,e., $x$ and $y$) lose an incident edge for input wires that
  evaluate to true (as the gates corresponding to these gadgets
  precede $g$ in the topological order), the above table then states that the gadget works
  correctly. The only pitfall we need to address is that the simulation
  does not ``work backwards'', that is, that reduction rules in $g$ trigger a
  reduction rule for the ``output'' vertex~$v$ of a gadget that
  corresponds to a gate that precedes $g$ in the topological
  order. The only way in which the described scenario appears is
  when $v$ obtains a self-loop (as $v$ is an ``output'' vertex it is
  connected to a clique of size three and, hence, the only rule that can delete
  $v$ is the Loop~Rule). The only way to generate a self-loop
  is to contract a cycle of degree-2 vertices to $v$. With out loss
  of generality, we may assume that $v$ is connected to exactly one of
  $x$ and $y$ and we may assume that it is $x$. Therefore, to generate
  a self-loop on $v$ the reduction rules have to modify the gadget
  such that $x$ has exactly two incident edges which both are connected
  to the vertex~$v$~--~the case distinction in the table shows that
  this can not happen.
  
  The induction completes the proof. A full example of the
  construction is provided in Figure~\ref{figure:fvsExample}. It is worth mentioning that the
  construction almost never generates vertices to which the Leaf Rule
  can be applied, the sole exception are \Lang{or}-gates in which
  only one input is set to true (after the deletion of the ``output''
  vertex, the second input vertex becomes a leaf). However, in this
  case the application of the Leaf Rule has no effect on the behaviour
  of the gadgets. Also note that the simulation still works if we
  alternatingly apply the Chain~Rule and the Loop Rule exhaustively
  (i.\,e., exhaustively apply the Chain Rule, exhaustively apply the
  Loop Rule, exhaustively apply the Chain Rule, and so forth). These
  observations yield Remark~\ref{remark:fvs:ruleP}.
\end{proof}
\begin{figure}[h]
  \begin{center}
    \includegraphics{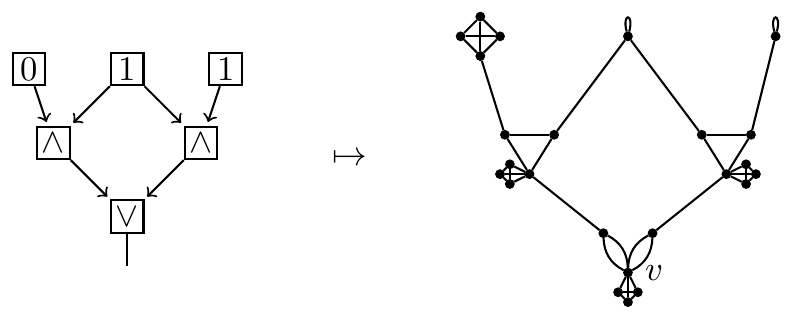}
  \end{center}
  \caption{An example of the construction from the proof of
    Theorem~\ref{theorem:fvs:ruleP}. The vertex labeled $v$ is the
  ``output'' vertex of the gadget corresponding to the output gate of
  the circuit, that is, $v$ gets removed by the reduction rules if, and
  only if, the circuit evaluates to true. This is the case in this
  example, but if we change, for instance, the third input gate to
  false (replace the self-loop with a clique), $v$ would not be removed.}
  \label{figure:fvsExample}
\end{figure}
\begin{bremark}\label{remark:fvs:ruleP}
  The proof of Theorem~\ref{theorem:fvs:ruleP} shows that the problem
  remains $\Class{P}$-hard restricted to the Chain Rule and the Loop
  Rule, even if they are alternatingly executed exhaustively.
\end{bremark}
We close this section with the observation that also the last rule,
the Flower Rule, is unlikely to yield a parallel algorithm.
\begin{theorem}\label{theorem:flowerRule}
  Unless $\Lang{matching}\in\Class{NC}$, there is no
  \textsc{dlogtime}-uniform family of $\Class{NC}^i$\-circuits for any~$i$ that
  determines, give a graph $G=(V,E)$, an integer~$k$, and a vertex~$v$,
  whether the Flower Rule can be applied to~$v$.
\end{theorem}
\begin{proof}
  Assume we have access to an $\Class{NC}$-circuit that determines if
  the Flower Rule can be applied to some vertex $v$. We construct an
  $\Class{NC}$-circuit for \Lang{matching}. On input $G=(V,E)$ and
  $k\in\mathbb{N}$, the circuit constructs a new graph $G'$ by adding
  a vertex~$s$ to~$G$, which is connected to all vertices of~$V$. It
  then uses the $\Class{NC}$-circuit for the Flower Rule on~$s$ and
  the value $k-1$. We claim $G$ has a matching of size~$k$ if, and only
  if, the Flower Rule can be applied to~$s$. For the first direction
  observe that, if $s$ is in $k$ disjoint cycles, we may take one edge of
  every cycle and, hence, have a matching of size $k$. Now assume $G$
  contains a matching of size $k$, say $M=\{\, (m_1,m_2),
  (m_3,m_4),\dots,(m_{2k-1},m_{2k})\,\}$. Then $s$ is contained in $k$
  disjoint cycles, namely 
  $s - m_i - m_{i+1}-s$ for $i\in \{1,3,5,\dots,2k-1\}$.
\end{proof}

\section{Parallel Kernels for Structural Parameterizations}\label{section:structural}

It is known that 
$\Class{NP}$-hard graph parameters that are closed under taking
disjoint union do not allow a polynomial kernel unless
$\Class{NP}\subseteq\Class{coNP}/\mathrm{poly}$~\cite{Cygan:2015fr}. Famous problems that suffer from this result are the decision versions of tree width, path width, and tree
depth, which has led to a growing body of research that
considers structural parameters for these
problems~\cite{BodlaenderJK13, BodlaenderJK12, KobayashiT16}. A
commonly used parameter in this line of research is the vertex cover
number of the input graph 
and in this section we extend the
cited results by proving that the corresponding kernels can be
computed in small circuit classes. 

We use the following standard definitions: A \emph{tree decomposition}
of a graph $G=(V,E)$ is a tuple $(T,\iota)$ where $T$ is a tree and
$\iota$ a mapping from the nodes of~$T$ to subsets of~$V$ (which we
call \emph{bags}) such that for every $u\in V$ and every $\{v,w\}\in
E$ there is (1) a node $n$ with $u\in\iota(n)$, (2) a node $m$ with
$\{v,w\}\subseteq\iota(m)$, and (3) the set $\{\,n \mid u \in \iota(n)\,\}$ is connected in~$T$. The
\emph{width} of a tree decomposition is the maximum size of the bags
minus one. For a graph~$G$, its \emph{tree width $\tw(G)$} is the
minimum width of all tree decompositions of~$G$, its \emph{path width $\pw(G)$}
is the minimum width of all tree decompositions of~$G$ that are paths,
and its \emph{tree depth $\td(G)$} is the minimum width of all tree
decompositions $(T,\iota)$ of~$G$ that can be rooted in such a way
that for all $n,m\in V(T)$ we have $\iota(n)\subsetneq\iota(m)$ if $m$ is an descendant of $n$.
The following two facts will be useful, where $N(v) = \{\,u
\mid \{u,v\} \in E\,\}$ is the neighborhood of~$v$, $N[v] = N(v) \cup
\{v\}$, and where we call a vertex $v$ \emph{simplicial} if $N(v)$ is a clique:

\begin{fact}[\cite{BodlaenderJK13, BodlaenderJK12, KobayashiT16}]\label{fact:nighborrule}
  Let $G=(V,E)$ be a graph with tree width, path width, or tree depth
  at most $k$ and with $u,v\in V$, $\{u,v\}\not\in E$, and
  $|N(u)\cap N(v)|>k$. Then adding the edge $\{u,v\}$ to~$G$ will not increase
  the tree width, path width, or tree depth of~$G$, respectively.
\end{fact}
\begin{fact}[\cite{Bodlaender96}]\label{fact:simplicial}
  Let $G=(V,E)$ be a graph and $v\in V$ be a simplicial vertex, then
  we have $\tw(G)\geq|N(v)|$.
\end{fact}

\subparagraph*{Computing a Kernel for Tree Width.}

For the problem
$\PLang[vc]{tree-width}$ we are given a graph $G=(V,E)$, an integer $k$,
and a vertex cover $S\subseteq V$ of $G$; the parameter is $|S|$ and
the question is whether $\tw(G)\leq k$ holds.
\begin{theorem}\label{theorem:twvs}
  There is a \textsc{dlogtime}-uniform family of
  $\Class{TC^0}$-circuits that, on
  inputs of a triple $(G,k,S)$, outputs a $\PLang[vc]{tree-width}$
  kernel with at most $O(|S|^3)$ vertices.
\end{theorem}
\begin{proof}
  On input $(G,k,S)$ the circuit can easily check if $S$ is actually a
  vertex cover and if we have $k<|S|$. If not, it outputs a trivial no-instance in the first
  case and a yes-instance in the second case (a tree decomposition
  of width $|S|$ can easily be obtained from~$S$).

  The circuit now checks in parallel for every pair $u,v\in S$ with
  $\{u,v\}\not\in E$ if we have $|N(u)\cap N(v)\cap (V\setminus S)|>k$, that is, if
  the two vertices have more than $k$ common neighbors in
  $V\setminus S$. If this is the case, the circuit adds the edge
  $\{u,v\}$. Note that this operation is safe by
  Fact~\ref{fact:nighborrule} and can be applied in parallel as
  we consider only neighbors in $V\setminus S$ while we only add edges
  in $S$.
  Finally, the circuit considers all simplicial vertices $v\in
  V\setminus S$ in
  parallel: if $|N(v)|>k$, the circuit safely outputs a trivial no-instance by
  Fact~\ref{fact:simplicial}, otherwise the circuit safely removes $v$ from
  the graph by standard arguments~\cite{BodlaenderJK12}.

  We now argue that, if the circuit has not decided yet, the remaining
  graph has at most $O(|S|^3)$ vertices: the remaining graph consist
  of the vertices in $S$, and the nonsimplicial vertices $I\subseteq
  (V\setminus S)$. We have $|I|\leq |S|^3$ as any vertex $u\in I$ must
  have at least two neighbors $v,w$ in $S$ with $\{v,w\}\not\in E$ (as
  otherwise $u$ would be simplicial), however, every pair of
  nonadjacent vertices in $S$ can have at most $k$ common neighbors
  (as otherwise the circuit would have added the edge). Since we have
  at most $|S|^2$ such pairs, the claim follows by $k\leq|S|$.
\end{proof}
\begin{corollary}\label{corollary:twvsac}
  For every $\delta\in\mathbb{N}$ there is a \textsc{dlogtime}-uniform family of
  $\Class{AC^0}$-circuits that, on
  input of a triple $(G,k,S)$, outputs a $\PLang[vc]{tree-width}$
  kernel with at most $2^{\sqrt[\delta]{|S|}}$ vertices.
\end{corollary}
\begin{corollary}
  $\PLang[vc]{tree-width}\in\Para\Class{AC}^0$.
\end{corollary}

\subparagraph*{Computing a Kernel for Path Width.}

We define the problem $\PLang[vc]{path-width}$ analogously
to $\PLang[vc]{tree-width}$ and the aim of this section is to reformulate
Theorem~\ref{theorem:twvs} in terms of path width. The main difference
is that we cannot simply delete simplicial vertices as this would,
for instance, eliminate trees completely. We can, however, use the
following weaker result:
\begin{fact}[\cite{BodlaenderJK12}]\label{fact:simplicialpw}
  Let $G=(V,E)$ be a graph, $k\in\mathbb{N}$, and $v\in V$ be a simplicial vertex. If the
  degree $|N(v)|$ of $v$ is $1$ and the neighbor of~$v$ has another degree-$1$
  neighbor, or if we have $2\leq|N(v)|\leq k$ and for each pair
  $x,y\in N(v)$ there is a simplicial vertex $w\in N(x)\cap N(y)$ with
  $w\not\in N[v]$, then $\pw(G)\leq k$ if, and only if, $\pw(G[V\setminus\{v\}])\leq k$.
\end{fact}
\begin{theorem}\label{theorem:pwvs}
  There is a \textsc{dlogtime}-uniform family of
  $\Class{TC^0}$-circuits that, on
  input of a triple $(G,k,S)$, outputs a $\PLang[vc]{path-width}$
  kernel with at most $O(|S|^3)$ vertices.
\end{theorem}
\begin{proof}
  The circuit works as in Theorem~\ref{theorem:twvs} and differs only
  in the last step, that is, the handling of simplicial vertices. 
  We have to identify the vertices for which
  Fact~\ref{fact:simplicialpw} applies in constant parallel time,
  which is not trivial since we have dependencies between these
  vertices. The circuit \emph{marks} simplicial vertices to which
  Fact~\ref{fact:simplicialpw} does not apply or which we will use as
  conditions when applying the fact to other vertices as follows:
  The circuit first marks for every $v\in S$ the lexicographically
  smallest degree-$1$ 
  neighbor of $v$. Then for every simplicial vertex $v\in V\setminus
  S$ of degree at least~$2$, the circuit marks for every pair of
  neighbors $x,y$ of $v$ the lexicographically smallest simplicial
  vertex $w\in (N(x)\cap   N(y))\setminus N[v]$. If for any pair such
  a vertex does not exist,   $v$ marks itself.  
  Note that all simplicial vertices that are not marked can safely be
  removed by Fact~\ref{fact:simplicialpw} and since, furthermore, the
  safeness is witnessed by marked vertices, the circuit can remove
  them all in parallel.

  We are left with the task to show that there are at most $O(|S|^3)$
  marked vertices left (the other vertices can be counted as in
  Theorem~\ref{theorem:twvs}). We have at most $|S|$ marked vertices
  of degree~$1$ (one for each vertex in $S$), and at most $|S|^2$
  marked vertices of degree greater than~$1$: each such vertex~$v$ has a
  pair of neighbors in~$S$ that has $v$ as sole simplicial neighbor.
\end{proof}
\begin{corollary}
  For every $\delta\in\mathbb{N}$ there is a \textsc{dlogtime}-uniform family of $\Class{AC^0}$-circuits
  that, on input of a triple
  $(G,k,S)$, outputs a $\PLang[vc]{path-width}$
  kernel with at most $2^{\sqrt[\delta]{|S|}}$ vertices.
\end{corollary}
\begin{corollary}
$\PLang[vc]{path-width}\in\Para\Class{AC}^0$
\end{corollary}

\subparagraph*{Computing a Kernel for Tree Depth.}

The last problem we consider is tree depth, and, as for path width,
we prove a version of Theorem~\ref{theorem:twvs} for it. The main problem is once more that we cannot simply remove
simplicial vertices. However, by the following fact of Kobayashi and
Tamaki there are still enough simplicial vertices that are safe to
remove:   
\begin{fact}[\cite{KobayashiT16}]\label{fact:tdsimplicial}
  Let $G=(V,E)$ be a graph, $k\in\mathbb{N}$, and let $v\in V$ be a simplicial vertex
  with $1\leq|N(v)|\leq k$.  If every neighbor of $v$ has degree at
  least $k+1$, then we have $\td(G)\leq k$ if, and only if, $\td(G[V\setminus\{v\}])\leq k$.
\end{fact}
\begin{theorem}\label{theorem:tdvs}
  There is a \textsc{dlogtime}-uniform family of
  $\Class{TC^0}$-circuits that, on
  input of a triple $(G,k,S)$, outputs a $\PLang[vc]{tree-depth}$
  kernel with at most $O(|S|^3)$ vertices.
\end{theorem}
\begin{proof}
  We proceed again as in Theorem~\ref{theorem:twvs} and only differ in
  the way we handle simplicial vertices. In particular, we argue
  how we can apply Fact~\ref{fact:tdsimplicial} in parallel constant
  time. The circuit starts by marking for every vertex $v\in S$ with
  $|N(v)|>k$ the $k+1$ lexicographically smallest neighbors of
  $v$, then the circuit marks every simplicial vertex $v\in V\setminus
  S$ that has at least one neighbor of degree less than $k$.
  Note that every simplicial vertex that is not marked can safely be
  removed by Fact~\ref{fact:tdsimplicial} and, since this safeness is
  witnessed by marked vertices, these vertices can be removed in
  parallel.

  The amount of remaining vertices can be computed as in
  Theorem~\ref{theorem:twvs}, we will end the proof by counting the
  number of marked vertices. There are at most $|S|^2+|S|$ marked
  vertices that were marked in the first step, as every vertex in $S$
  marks only $k+1$ neighbors. Additionally we may have some simplicial
  vertices that are marked because they have a neighbor of degree at
  most $k$. Since every degree $k$ vertex in $S$ can produce at most
  $k$ such vertices, the number of these vertices can be bounded by
  $|S|^2$ as well.
\end{proof}
\begin{corollary}
  For every $\delta\in\mathbb{N}$ there is a \textsc{dlogtime}-uniform family of
  $\Class{AC^0}$-circuits that, on
  inputs of a triple $(G,k,S)$, outputs a $\PLang[vc]{tree-depth}$
  kernel with at most $2^{\sqrt[\delta]{|S|}}$ vertices.
\end{corollary}
\begin{corollary}
$\PLang[vc]{tree-depth}\in\Para\Class{AC}^0$
\end{corollary}

\section{A Parallel Kernel for Point Line Cover}\label{section:tc0}

In this section we study a natural, well-known
problem for which we can prove (unconditionally) that we \emph{cannot}
compute a kernel using $\Class{AC}^0$-circuits while we \emph{can}
compute polynomially-sized kernels in $\Class{TC}^0$. 
In the $\PLang{point-line-cover}$ problem we are given distinct points
$p_1,\dots,p_n\in\mathbb{Z}^d$ for some dimension $d \ge 2$ and a
natural number $k\in\mathbb{N}$, the question is whether we can cover
all points by at most $k$ lines. This problem is $\Class{NP}$-hard in
general (even for $d=2$) and in $\Class{FPT}$ parameterized by
$k$~\cite{KratschPR16}.  There is a simple $k^2$ kernel, which is
essentially optimal~\cite{KratschPR16}: If any line covers at least
$k+1$ points, remove all points on this line and reduce $k$ by
one. This is safe since we would require at least $k+1$ different
lines if we would not use this line. Because no set of $k+1$ points lies
on the same line after the reduction, we have at most $k^2$ points
left or we deal with a no-instance.
\begin{lemma}\label{lemma:plc}
  There is a \textsc{dlogtime}-uniform family of
  $\Class{TC^0}$-circuits that, on
  input of a dimension~$d$, a set of distinct points
  $p_1,\dots,p_n\in\mathbb{Z}^d$, and an integer $k$, outputs a
  $\PLang{point-line-cover}$ kernel with at most $k^2$ points.
\end{lemma}
\begin{proof}
  First observe that the reduction rule ``for a line covering at least 
  $k+1$ points, remove all points on this line and reduce $k$ by~$1$''
  can be applied in 
  parallel, as removing all points from a line removes at most one
  point from any other line.  To complete the proof, note that it is
  sufficient to check all $n^2$ line segments defined by pairs of points in parallel; and that
  a $\Class{TC}^0$-circuit can check if another point lies on such a
  line segment as it can multiply and divide binary
  numbers~\cite{Hesse01}.
\end{proof}
The lemma shows that the optimal kernel for $\PLang{point-line-cover}$
can be computed in $\Class{TC}^0$ and it is natural to ask if we can
do the same using a $\Class{AC}^0$-circuit or, failing that, to at
least compute \emph{some} kernel using  a $\Class{AC}^0$-circuit (as
we could for the problems in the previous sections).  We answer this
question in the negative, settling the complexity of the problem to
$\Para\Class{TC}^0$:
\begin{lemma}\label{lemma:plctcc}
  For every fixed~$k$, the $k$th slice of $\PLang{point-line-cover}$
  is $\Class{TC}^0$-complete under $\Class{AC}^0$-reduction.
\end{lemma}
\begin{proof}
  We start with the case $k=1$ and $d=2$, which is clearly in
  $\Class{TC}^0$, as an instance is a yes-instance if, and only if,
  the input points are colinear. To see that the problem is
  $\Class{TC}^0$-hard we reduce 
  from $\Lang{division}$ defined as: Given three numbers $x$, $y$,
  and~$z$, is it true that $x/y = z$? This is a classical 
  $\Class{TC}^0$-complete problem~\cite{Hesse01}. For the reduction
  let $x$, $y$, $z$ be the $\Lang{division}$-instance, we construct
  the instance $a=(0,0)$, $b=(x,z)$, $c=(y,1)$ of
  $1\Lang{-point-line-cover}$. This is a yes-instance if the points
  are colinear, that is, if we have $(b-a)\cdot(c-a)=0$ or,
  equivalently: $ \frac{x-0}{y-0}=\frac{z-0}{1-0}\iff x/y=z$.
  Since the cases $k > 1$ and\,/\,or $d > 2$ are generalizations, they
  remain $\Class{TC}^0$-hard. To see that these cases are also in
  $\Class{TC}^0$, observe that we have to consider at most $n^2$ line segments
  from which we have to pick~$k$, that is, there are at most
  $\binom{n^2}{k}\leq n^{2k}$ solution candidates. For fixed $k$, these
  candidates can be checked in parallel by a $\Class{TC}^0$-circuit and
  can be evaluated as in the case of $k=1$.
\end{proof}
\begin{corollary}
$\PLang{point-line-cover}$ is $\Para\Class{TC}^0$-complete under $\Class{AC^0}$-reduction.
\end{corollary}
Now assume there would be a uniform family of $\Class{AC}^0$-circuits
computing a kernel of arbitrary size for
$\PLang{point-line-cover}$. Then by Theorem~\ref{theorem:kernel} the
problem is in $\Para\Class{AC}^0$, which on the other hand implies
that for every fixed $k$ the problem must be in $\Class{AC}^0$. This
contradicts Lemma~\ref{lemma:plctcc} as it is known that
$\Class{AC}^0\subsetneq\Class{TC}^0$~\cite{FurstSS84}. Therefore, no family of
$\Class{AC}^0$-circuits can compute such a kernel. 
\begin{corollary}
$\Para\Class{AC}^0\not\ni\PLang{point-line-cover}\in\Para\Class{TC}^0$
\end{corollary}

\section{Problems for Which Computing a Kernelization is Inherently Sequential}\label{section:seq}

As surprisingly many problems have $\Class{NC}$-computable, in
fact often even $\Class{AC}^0$-computable, kernelizations, we may ask
which problems do not have this
property. We would like to find problems for which the computation of
any kernel is $\Class{P}$-complete or, equivalently, 
which are $\Class{FPT}$-complete under
$\Class{AC}^0$- or $\Class{NC^1}$-reductions. While it is easy to find
artificial problems with this property -- such as any
$\Class{P}$-complete problem (like $\Lang{cvp}$) with the trivial
parametrization ($\kappa(x) \equiv 1$) --, no \emph{natural} problems
that are $\Class{FPT}$-complete for sensible parametrizations can be
found in the literature. We remedy this situation in the following;
but must caution the reader that in all our results the hardness of
the parameterized problem for $\Class{FPT}$ 
stems from the fact that some slice of the problem is (essentially) a
known $\Class P$-complete problem. Unfortunately, it is known~\cite{FlumG06}
that this ``cannot be helped'' since all $\Class{FPT}$-complete
problems have this property. Our main contribution here lies, thus, in
the assembly of a diverse body of relevant, non-trivial $\Class{FPT}$-problems
that will serve as starting points for further studies of the limits
of parameterized parallelization. 

\subparagraph*{Strong Backdoors to Satisfiability.}  A \emph{strong
  backdoor set} of a propositional formula $\phi$ is a set of
variables such that under any assignment of these variables the
resulting formula $\phi'$ belongs to a certain class of formulas~\cite{GaspersS12}. In the
$\PLang{strong-backdoor-\{horn,2cnf\}-sat}$ problems, we are given a
formula~$\phi$ and an integer~$k$, the question is whether $\phi$ is
satisfiable and has
a strong backdoor set of size~$k$ to Horn- or 2\textsc{cnf}-formulas,
respectively. Solving such
problems is usually done in two phases: first \emph{detect} the
backdoor set and, second, \emph{solve} the satisfiability problem of
the formula for every assignment of the backdoor set. While the first
part might seem harder in general, it is not from a parameterized
point of view: (1) A strong backdoor set to Horn
formulas is exactly a vertex cover of size $k$ in the positive primal
graph of~$\phi$, that is, the graph that has a vertex for each
variable and an edge between any two variables appearing together
positively in a clause; (2) strong backdoor sets to
\textsc{2cnf}-formulas are exactly the hitting sets of the hypergraph
that has the variables of $\phi$ as vertices and that connects three
vertices by a hyperedge if they appear together in a clause. Since
$\PLang{vertex-cover}\in\Para\Class{AC}^0$ and also
$\PLang{3-hitting-set}\in\Para\Class{AC}^0$~\cite{BannachST15,
  ChenFH17}, we can conclude:

\begin{corollary}\label{corollary:backdoor}
  There is a \textsc{dlogtime}-uniform family of
  $\Para\Class{AC^0}$-circuits that, on input of a propositional formula~$\phi$ and an integer~$k$,
  either outputs a size-$k$ strong backdoor set to $\{\text{Horn},
  \textsc{2cnf}\}$-formulas, or concludes that no such set exists.
\end{corollary}

The second step of solving
$\PLang{strong-backdoor-\{horn,2cnf\}-sat}$ is to solve the
satisfiability problem for $\phi$ on
every assignment to the variables of the backdoor set. While we can nicely
handle all assigments in parallel, checking if the
formulas are satisfiable in parallel is difficult. Indeed, it is known
that, under $\Class{AC}^0$-reductions,
the satisfiability problem is $\Class{NL}$-complete for
\textsc{2cnf}-formulas, and is even $\Class{P}$-complete for Horn
formulas~\cite{AllenderBISV09}.
\begin{corollary}
  $\PLang{strong-backdoor-2cnf-sat}$ is $\Para\Class{NL}$-complete
  under $\Class{AC}^0$-reduction.
\end{corollary}
\begin{corollary}\label{corollary:hronfpt}
  $\PLang{strong-backdoor-horn-sat}$ is $\Class{FPT}$-complete under
  $\Class{AC}^0$-reduction.
\end{corollary}
The last corollary implies that there is no parallel kernelization
running in polylogarithmic
time for $\PLang{strong-backdoor-horn-sat}$ that
produces a kernel of any size, unless $\Class{NC}=\Class{P}$.

\subparagraph*{Mixed Integer Linear Programming.} The
$\Class{FPT}$-complete problem above is an intermediate problem
between a $\Class{P}$-complete problem 
(\Lang{horn-sat}) and a $\Class{NP}$-complete problem (\Lang{sat});
the transition between the problems is caused by the backdoor
variables. A similar intermediate problem is known for
\Lang{linear-programming} (another classical $\Class{P}$-complete
problem) and its integer variant (which is $\Class{NP}$-complete). The
intermediate version of these problems is called
$\PLang{mixed-integer-programming}$, which asks,
given a matrix $A\in\mathbb{Z}^{n\times n}$, vectors
$b\in\mathbb{Z}^{n}$, $c\in\mathbb{Z}^{n}$, and
integers $k$ and $w$, if there is a vector
$x\in\mathbb{R}^{n}$ such that $Ax\leq b$, $c^{\mathrm{T}}x\geq w$, and
such that $x[i]\in\mathbb{Z}$ for $0\leq i< k$. A celebrated result by Lenstra states that an
instance $I$ of this problem can be solved in time
$2^{O(k^3)}\cdot |I|^c$ for a suitable constant $c$, that is, the problem is in
$\Class{FPT}$. Therefore, every slice of the problem is in $\Class{P}$
and, as \Lang{linear-programming} trivially reduced to it, we get that
$k$-\Lang{mixed-integer-programming} is $\Class{P}$-complete for every
$k$ (under $\Class{NC}^1$-reductions~\cite{toran1993p}).
\begin{corollary}
  $\PLang{mixed-integer-programming}$ is $\Class{FPT}$-complete
  under $\Class{NC}^1$-reductions.
\end{corollary}

\subparagraph*{Maximum Flow with Minimum Quantities.}
The last problem we review in this section is the maximum flow problem
with minimum quantities: Inputs are directed graphs $G=(V,E)$ with
$s,t\in V$, two weight functions $u,l\colon E\to\mathbb{N}$, an
integer $w\in\mathbb{N}$, and a set of edges $B\subseteq E$; the
question is whether there is a set $A\subseteq B$ such that in
$G'=(V,E\setminus A)$ there is a valid $s$-$t$-flow $f$ of value at least $w$
that fulfills the flow conservation constraints and $l(e)\leq
f(e)\leq u(e)$ for all $e\in E\setminus A$. For $B=\emptyset$ the
problem boils down to classical maximum flow with lower bounds on the
edges, which can be solved in polynomial time~\cite{kleinbergT2006} and which is
known to be $\Class{P}$-hard under
$\Class{NC^1}$-reduction~\cite{toran1993p}. On the other hand, for
$B=E$ the problem becomes $\Class{NP}$-complete even on
serial-parallel graphs~\cite{HauglandEH11} and it is also
$\Class{NP}$-hard to 
approximate the problem within any positive
factor~\cite{ThielenW13}. The intermediate problem between this two
cases is the parameterized problem $\PLang{max-flow-quantities}$  where the cardinality of $B$ is the parameter.
\begin{lemma}\label{lemma:mfq}
  $\PLang{max-flow-quantities}$ is $\Class{FPT}$-complete under
  $\Class{NC^1}$-reduction.
\end{lemma}
\begin{proof}
  Containment in $\Class{FPT}$ follows by the simple algorithm that
  iterates over all $2^{|B|}$ possible sets $A\subseteq B$ and which
  computes a maximum flow in $G'=(V,E\setminus A)$ with, for instance,
  a variant of Ford-Fulkerson~\cite{kleinbergT2006}.
  The algorithm also implies that for every fixed $k$ the slice
  $k$-\Lang{max-flow-quantities} is in $\Class{P}$ and, since it is a
  generalization of classical $\Lang{max-flow}$, it is also $\Class{P}$-complete.
\end{proof}

\section{Conclusion and Outlook}\label{section:conclusion}

Kernelization is a fundamental concept of parameterized complexity and
we have studied its parallelization.
Since traditional descriptions of kernelization algorithms are inherently
sequential, we found it surprising how many parameterized problems
lie in $\Para\Class{AC}^0$ -- the smallest robust class in parallel
parameterized complexity theory. We found, furthermore, that for many
problems the equation ``smaller circuit class = larger kernel'' holds,
see Table~\ref{table:overview} for a summary of our results.

Apart from classifying more parameterized problems in the spirit of
this paper, namely according to how well small kernels can be computed
by small circuits, an interesting open problem is to improve any of
the $\Class{AC}^0$-kernelizations presented in the paper so
that they produce a \emph{polynomially sized} kernel (which we, at
best, can
currently do only in $\Class{TC}^0$). Perhaps even more challenging seems
to be the design of a framework for proving that polynomially sized
kernels for these problems can\emph{not} be computed in~$\Class{AC}^0$.

%%
%% Bibliography
%%
\bibliography{main}

\end{document}